\documentclass[11pt]{article}
\usepackage{dsfont,amsthm,amsmath,marvosym}
\usepackage{amssymb}
\usepackage{mathptm}
\usepackage{color}
\usepackage[margin=1in]{geometry}
\usepackage{graphicx,enumitem}
\usepackage{cite}
\usepackage{hyperref}
\usepackage{mathrsfs}
\newif\ifFull 
\Fulltrue

\newtheorem{lemma}{Lemma}
\newtheorem{theorem}{Theorem}
\newtheorem{corollary}{Corollary}
\newtheorem{claim}{Claim}

\newcommand{\Z}{\ensuremath{\mathds{Z}}}

\newcommand{\ssg}{\ensuremath{\mathrm{sep}}}
\newcommand{\snip}{\mathbin{\text{\raisebox{0.15ex}{\rotatebox[origin=c]{60}{\Rightscissors}\!}}}}
\newcommand{\IDDG}{\ensuremath{\mathrm{IDDG}}}
\newcommand{\DDG}{\ensuremath{\mathrm{DDG}}}

\renewcommand{\paragraph}[1]{\vspace{1mm} \noindent {\bf #1} \ }

\title{All-Pairs Minimum Cuts in Near-Linear Time\\ for Surface-Embedded Graphs}
\author{Glencora Borradaile\footnote{School of Electrical Engineering and Computer Science, Oregon State University.} \and David Eppstein \footnote{Computer Science Department, Donald Bren School of Information and Computer Sciences, University of California, Irvine.} \and Amir Nayyeri\footnotemark[1] \and Christian Wulff-Nilsen\footnote{Department of Computer Science, University of Copenhagen (DIKU).}}

\date{ }

\begin{document}
\maketitle
\thispagestyle{empty}

\begin{abstract}
For an undirected $n$-vertex graph $G$ with non-negative edge-weights, we
consider the following type of query: given two vertices $s$ and $t$ in $G$,
what is the weight of a minimum $st$-cut in $G$? We solve this problem in preprocessing time $O(n\log^3 n)$ for graphs of bounded genus, giving
the first sub-quadratic time algorithm for this class of graphs. Our result also improves by a logarithmic factor a previous algorithm by Borradaile, Sankowski and
Wulff-Nilsen (FOCS 2010) that applied only to planar graphs. Our algorithm constructs a Gomory--Hu tree for the given graph, providing a data structure with space $O(n)$ that can answer minimum-cut queries in constant time.  The
dependence on the genus of the input graph in our preprocessing time is~$2^{O(g^2)}$.
\end{abstract}
\newpage
\setcounter{page}{1}

\section{Introduction}

In the {\em all-pairs minimum cut problem} we seek the minimum $st$-cut for every pair $\{s,t\}$ of vertices in an edge-weighted, undirected graph $G$. Gomory and Hu~\cite{GH61} showed that these minimum cuts can be represented by a single edge-weighted tree such that:
\begin{itemize} [noitemsep,nolistsep]
\item the nodes of the tree correspond one-to-one with the vertices of $G$,
\item for any distinct vertices $s$ and $t$, the minimum-weight edge on the unique $s$-to-$t$ path in the tree has weight equal to the min $st$-cut in $G$, and
\item removing this minimum-weight edge from the tree creates a partition of the nodes into two sets corresponding to a min $st$-cut in $G$.
\end{itemize}
We call such a tree a minimum cut tree; it is also known as a Gomory--Hu tree or cut-equivalent tree. Gomory and Hu showed how to find this tree with $n-1$ calls to a minimum cut algorithm by building up a collection of nested cuts, in each step adding a minimum $st$-cut that separates a previously-unseparated pair of vertices.

\paragraph{New results} We provide the first subquadratic algorithm for all-pairs minimum cuts in bounded-genus graphs. We can find the Gomory--Hu tree of a graph of genus $g$ in time $2^{O(g^2)}n\log^3 n$, giving a data structure of size $O(n)$ which can answer minimum-cut queries in constant time.
The  best previous method for this class of graphs uses the standard Gomory--Hu algorithm and has a running time of $O(g^8n^2 \log^2n \log^2C)$ (for integer edge weights summing to $C$) using the best maximum-flow algorithm to find minimum cuts~\cite{CEN12} or $2^{O(g)}n^2 \log n$ using the best minimum-cut algorithm~\cite{EN11} for graphs of bounded genus.
Our result hinges in part on an improvement in the time for constructing Gomory--Hu trees in planar graphs. Borradaile, Sankowski and Wulff-Nilsen showed how to solve this problem in $O(n \log^4 n)$ time~\cite{BSW14}.  In this paper, we improve this running time to $O(n \log^3 n)$ time (\autoref{sec:planar-speed-up}).

\paragraph{From planar to bounded genus} We reduce the problem of computing the minimum-cut tree in a graph of genus $g$ to the same problem in a set of $2^{O(g^2)}$ planar graphs.  
The minimum cut, viewed as an even subgraph of the dual graph, is a collection of cycles belonging to one of $2^{2g}$ homology classes.
Our main observation is that one can reduce the problem of finding a minimum cut that is composed of dual cycles in certain homology classes to a planar problem {\em before} taking into account the vertices that should be separated.  This allows us to find a cut tree whose cuts are composed of dual cycles in certain homology classes.  We describe this reduction in \autoref{sec:g}.  Through this reduction, we compute minimum cut trees in $2^{O(g^2)}$ different planar graphs such that the minimum $st$-cut in the original graph is represented in at least one of these $2^{O(g^2)}$ cut trees. Although this would already solve the minimum cut query problem, we also show
\ifFull
\autoref{sec:merge}
\else
in the full version of this paper
\fi
how to produce a single minimum cut tree for the original graph, by merging these cut trees in a way that preserves minimum cuts in time $O(k n\log^2 n)$ where $k$ is the number of cut trees to merge.  This tree-merging algorithm does not rely on the surface embedding.  Note that one could use the Gomory-Hu algorithm of Gusfield to merge these trees, but Gusfield's algorithm has an $O(n^2)$ overhead independent of the method for finding minimum cuts~\cite{Gusfield90}.

\paragraph{Planar speed-up} The algorithm of Borradaile et~al.~\cite{BSW14}  for planar graphs implements the Gomory--Hu algorithm by using Miller's recursive cycle separator decomposition~\cite{Miller86} to guide the selection of pairs of vertices to separate.  Starting with the leaf-most pieces of the decomposition, the algorithm separates all pairs of vertices in a piece that are not yet separated.  We refer to this algorithm as the {\em cycle-based} algorithm as it works in the dual graph finding minimum separating cycles of pairs of faces.  

We improve the running time of this algorithm by addressing two bottlenecks.  
The first bottleneck is in finding a separating cycle. The cycle-based algorithm incurs an $O(\log^3 n)$ factor per cycle; multiplied by the depth of the recursive decomposition of the planar graph this results in an $O(\log^4 n)$ factor in the overall runtime.  Instead, we find a minimum cut by way of first computing a maximum flow, using the recursive flow techniques~\cite{BKMNW11,LNSW12}, which, surprisingly (because max flow computations usually dominate minimum cut computations), reduces the overhead per cycle to $O(\log^2 n)$.  We refer to this algorithm as the {\em flow-based} algorithm (\autoref{sec:st}).
The second bottleneck is in adding a cut to the collection; we improve the overhead from an amortized $O(\log^4 n)$ per cut to an amortized $O(\log^3 n)$ per cut (\autoref{sec:rt}).

\paragraph{Minimum cycle bases}   By duality~\cite{HM94}, our algorithm also finds the cycle weights in a minimum cycle basis of a planar graph in the same time bound; in fact, the planar all pairs min cut algorithm is really a minimum cycle basis algorithm in the dual graph.  However, in graphs of higher genus and even in toroidal graphs, the minimum cycle basis appears to be mostly unrelated to the dual of minimum cuts.  Chambers et al.~\cite{BCFN15} study computing the minimum cycle bases and the minimum homology basis on surface embedded graphs.

\section{Preliminaries}

We consider a graph $G$ with $n$ vertices with a cellular embedding on an orientable surface of genus $g$.\footnote{An embedding is cellular if every face is a topological disk. If $G$ has any embedding on a surface of genus~$g$, then the rotation system of the embedding gives a cellular embedding of $G$ on a surface of genus at most~$g$.}  In this paper, we use cycle in its topological sense, that is a closed curve;  graph theoretically, a cycle in this sense corresponds to a closed walk.  We specify simple cycle when we refer to a closed curve that visits each point (vertex, edge) at most once.

\paragraph{Duality} For every connected, surface embedded graph $G$ (the {\em primal}) there is another connected graph embedded on the same surface, the {\em dual} $G^*$. The faces of $G$ are the vertices of $G^*$ and vice versa. The edges of $G$ correspond one-for-one with the edges of $G^*$: for each edge $e$ in $G$, there is an edge $e^*$ in $G^*$ whose endpoints correspond to the faces of $G$ incident to $e$.  Dual edges inherit the weight of the corresponding primal edges; namely, $w(e^*) = w(e)$.

\paragraph{Surgery} For a cycle $C$ we define the operation of \emph{cutting} along $C$ in $G$ and denote it $G \snip C$.
$G \snip C$ is the graph obtained by cutting along $C$ in the drawing of $G$ on the surface, creating two copies of every edge in $C$.  The edges in the copies of $C$ inherit the weights of the original edges.  We view $G \snip C$ as being embedded on a surface with two punctures, corresponding to the two resulting copies of $C$.  

\paragraph{$\mathds{Z}_2$-homology} $\mathds{Z}_2$ homology as we use it in this paper is described by Erickson and Nayyeri~\cite{EN11};  we refer the reader to their paper for formal definitions of the following.  For further background on surface topology and homology we refer the reader to Hatcher~\cite{H02}. 
Here, when we talk about homology we mean homology with $\mathds{Z}_2$ coefficients.

A subgraph is called \emph{even} if it has even degree at every vertex, or equivalently if it is the edge-disjoint union of simple cycles.  An even subgraph is \emph{null-homologous} if it is the boundary of a union of faces in~$G$.  Two even subgraphs are homologous if their symmetric difference is null-homologous.  A \emph{homology basis} is a set $\{C_1, C_2, \ldots, C_{2g}\}$ of simple cycles in $G^*$ that generates the homology class of all cycles of $G^*$; it can be constructed in linear time~\cite{EN11}.  The \emph{signature} of an edge $[e]$ is defined as a $2g$-bit vector, whose $i$th bit is $1$ if and only if $e^* \in C_i$.  Given two faces $a$ and $b$ and an $a^*$-to-$b^*$ path $P$ in $G^*$, $[e]^{ab}$ is the extended $(2g+1)$-bit vector whose first bit is $1$ if and only if $e^*\in P$.  Note that $[e]^{ab}$ depends on the choice of $P$.  For any subgraph $X\in G$ we define $[X] = \bigoplus_{e\in X}{[e]}$ and $[X]^{ab} = \bigoplus_{e\in X}{[e]^{ab}}$.  Two even subgraphs $X$ and $X'$ are homologous if and only if $[X\oplus X'] = [X]\oplus [X'] = 0$.  We use $\oplus$ to denote the symmetric difference of sets and the exclusive or of binary numbers. 

\newcommand{\homsigfigure}{%
\begin{figure}[h]
  \centering
    \includegraphics[height=1.15in]{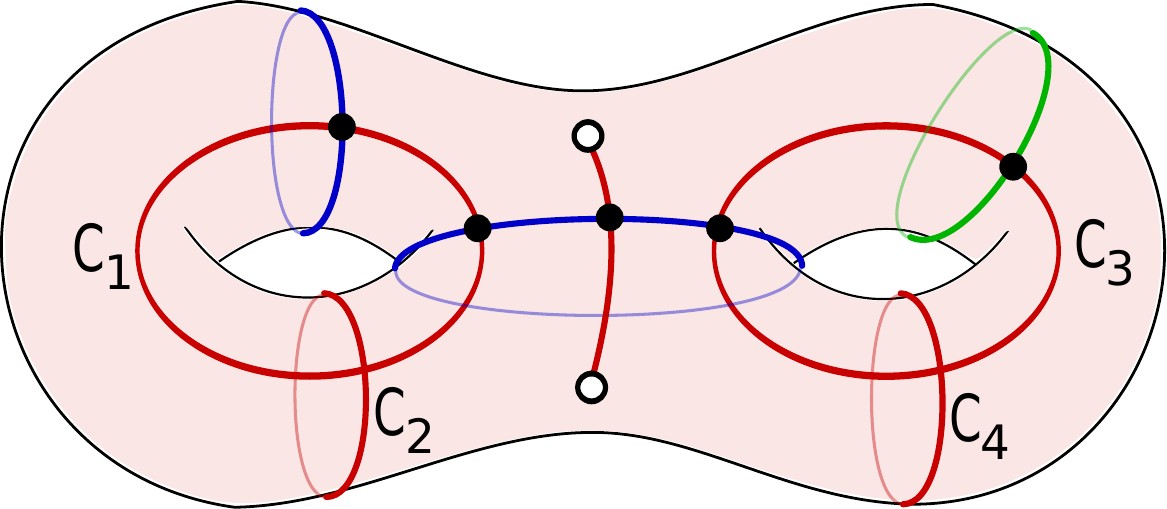}
  \caption{The $\mathds{Z}_2$ homology signature of an even subgraph.  The red cycles (in the dual) form a homology basis.  The red path (also in the dual) connects two vertices $a^*$ and $b^*$.  The disconnected blue even subgraph $X_B$ (in the primal) has signature $[X_B]^{ab} = (10010)$  and the green even subgraph $X_G$ (also in the primal) has signature $[X_G]^{ab} = (00010)$.}
  \label{fig:sig}
\end{figure}}

\ifFull
\homsigfigure
\fi

\paragraph{Minimum cuts and minimum separating subgraphs} The dual of a minimum $st$-cut is the minimum even subgraph $X$ that is null-homologous and such that $[X]_0^{s^*t^*} = 1$, i.e.\ $[X]^{s^*t^*} = [1\, 0 \cdots 0]$ (Lemma 3.1~\cite{CEN09}). We will call this the minimum $s^*t^*$-separating subgraph. In a surface of genus $g$, a minimum separating subgraph is composed of at most $g+1$ simple cycles.  In particular, a minimum separating subgraph in a planar graph is a simple cycle~\cite{Whitney32}. The all-pairs minimum cut problem is equivalent to the all-pairs minimum separating subgraph problem in the dual.

\paragraph{Faces and boundaries}  For a set of faces $F$, we define $\partial F$ to be the boundary of $F$, that is, the set of edges that bound faces of $F$ and $\bar F$.  Finally, for any graph $G$, we use $F(G)$ to denote the set of faces of $G$.
We additionally consider a special type of faces, {\em boundary} faces, which correspond to the boundary of punctures in the surface and are introduced over the course of our algorithm from the cutting along operation.  If $H = G \snip X$ for some set of edges $X$, then $F(H)$ is $F(G)$ plus a set of additional boundary faces, bounded by edges of $X$.

\begin{claim}\label{clm:sep}
    Let $X$ be a set of edges such that $G \snip X$ has one more boundary face than $G$.  Let $S$ be an $ab$-separating subgraph in $G \snip X$ for non-boundary faces $a,b \in F(G \snip X)$.  Then $S$ is $ab$-separating in $G$.
\end{claim}

\begin{proof}   
    Let $f$ be the boundary face that is created by cutting along $X$ in $G$.  Let $A,B$ be the bipartition of faces of $G \snip X$ given by the $ab$-separating subgraph $S$; w.l.o.g., take $a \in A$ and $b,f \in B$.  Since $G\snip X$ has just one additional boundary face  ($f$) than $G$, $A$ is a set of faces in $G$.  It follows that $S$ is $ab$-separating in $G$.
\end{proof}

\paragraph{Tight cycles and paths} We say that a cycle $C$ is \emph{tight} if it is the shortest cycle with $\mathds{Z}_2$-homology signature $[C]$.  Tight cycles in all homology classes can be found in time $2^{O(g)}n \log n$ time~\cite{EN11}.  We say that an $x$-to-$y$ path $P$ is tight if $P \cup xy$ is the shortest cycle with signature $[P \cup xy]$ in the graph $G \cup xy$ where $xy$ is embedded on a handle added to the surface connecting a face incident with $x$ to a face incident with $y$.  Note that the faces may coincide, and we may need to add extra edges (with large weights to ensure they are not part of any tight cycle) to make the embedding cellular.  In this way, a tight path can be found in the same time bound as a tight cycle.

\paragraph{Crossing or non-crossing}  Let $H_1$ and $H_2$ be two subsets of edges.  We say that $H_1$ \emph{crosses} $H_2$ if there is a subset of edges $S$ of $H_1 \cap H_2$ such that contracting $S$ results in a vertex $s$ such that the edges $e_1, e_2, e_1',e_2'$ are incident to $s$ and are in this clockwise order around $s$ with $e_1, e_1' \in H_1$ and $e_2,e_2' \in H_2$.  Otherwise, $H_1$ and $H_2$ do not cross.

If $H_1$ is connected and $H_2$ is an even separating subgraph, then the following is an equivalent definition.  Let $S_1, S_2, \ldots, S_k$ be the components of $G \snip H_2$.  $H_1$ crosses $H_2$ if and only if there are two edges $e$ and $f$ in $H_1$ such that the two faces incident to $e$ are in $S_i$ and the two faces incident to $f$ are in $S_j$ for $i \ne j$.

A cycle $C$ is {\em non-self-crossing} if no two subpaths of $C$ cross.  A cycle is {\em weakly simple} if it is non-self-crossing and traverses each edge at most once.  A {\em cycle decomposition} of an even subgraph $H$ is a partition ${\cal C} = \{C_1, C_2, \ldots \}$ of the edges of $H$ such that each cycle in $\cal C$ is simple and no two cycles in $\cal C$ cross.  Chambers, Erickson and Nayyeri (Lemma 3.2~\cite{CEN09}) prove that every even subgraph of a surface-embedded graph has a cycle decomposition.

\paragraph{Cuts} By abuse of notation, we use an $xy$-cut to refer to two equivalent notions: it can be either a subset $C$ of edges whose removal from the graph separates $x$ and $y$, or a bipartition $(X,Y)$ of the vertices such that $x\in X$ and $y\in Y$. The edge subset $C$ is the set of edges with one endpoint in $X$ and the other in $Y$.  We say that a cut $(A,B)$ {\em crosses} cut $(X,Y)$ if neither $A \subseteq X$ nor $A \subseteq Y$.

\paragraph{Uniqueness of minimum cuts and shortest paths} We assume that minimum cuts  do not cross. That is, if $(X,Y)$ is a minimum $xy$-cut and $(A,B)$ is a minimum $ab$-cut, then $(X,Y)$ and $(A,B)$ do not cross.  This is automatically true when all minimum cuts are unique, which, in turn can be assumed to be true with high probability by randomly perturbing the edge weights slightly~\cite{BSW14,MVV87}.  This perturbation also allows us to assume that shortest paths in the dual graph are unique, which is required for both the cycle- and flow-based planar algorithms.

\paragraph{Cut trees} For a graph $G$, an edge-weighted tree $T$ on the same vertex set as $G$ is a cut-tree for $G$ if, for every edge $e$ in $T$, the weight of the cut in $G$ corresponding to the bipartition of the vertices given by $T \setminus \{e\}$ is $w(e)$.  $T$ is a {\em minimum} cut-tree of $G$ if for every pair of vertices $x,y$ the minimum $xy$-cut in $T$ is the same (in value and bipartition of vertices) as the minimum $xy$-cut in $G$.

\paragraph{Region trees} Within our algorithms, it is convenient to represent a cut tree $T$ with a \emph{region tree} $R$~\cite{BSW14}.  $R$ is the unique tree obtained from $T$ by adding a leaf vertex to every node of $T$. Thus, $R$ has one leaf and one internal node for each of the $n$ vertices of $T$. The internal edges of $R$ are exactly the edges of $T$.  Such a region tree is called a {\em complete} region tree.  The cut tree $T$ can be recovered from $R$ by contracting all the leaf edges.

A {\em partial} region tree is any tree that can be obtained from a region tree by contracting a subset of internal edges.  A region tree may refer either to a complete region tree or to a partial region tree, and should be clear from context.  Contracting the leaf edges of a partial region tree and mapping subsets of $V$ to the resulting nodes gives a representation of a {\em partial} cut tree as maintained throughout the standard Gomory--Hu algorithm.

If a region tree $R$ is rooted at an arbitrary internal node, then its internal nodes represent the {\em regions} on one side of each cut of $T$. For a non-root, internal node $v$ of $R$, we define $S_v$ to be the set of leaf descendents of $v$; $S_v$ is one side of the cut $T \setminus \{e\}$ where $e$  is the parent edge of $v$.  The root of $R$ is the region that corresponds to the entire graph.

\paragraph{Cartesian trees} Our query data structure is based on a \emph{Cartesian tree} for an edge-weighted tree $T$. This is a binary tree in which the interior nodes represent edges of $T$ and the leaves represent vertices of $T$. The root node represents the lightest edge of $T$ and its two children are constructed recursively from the two subtrees formed from $T$ by removing this lightest edge. The minimum cut between any two vertices in $T$ can then be found by answering a lowest common ancestor query between the corresponding two leaves of the Cartesian tree. Given an edge-weighted tree~$T$ whose edges have been sorted by weight, the Cartesian tree for $T$ can be constructed and processed for constant-time lowest common ancestor queries in time and space $O(n)$~\cite{DLW14}.

\section{Reduction from bounded genus to planar}\label{sec:g}

We show how to reduce the all-pairs minimum cut problem for a surface-embedded graph $G$ to the planar case.  We do so by recursively cutting along (i) a tight cycle that belongs to some minimum cut or (ii) a tight cycle and a tight path connecting sides of the tight cycle that some minimum cut does not cross.

\subsection{Reducing the genus}

If $S$ is a minimum $ab$-separating subgraph in graph $G$ and $C$ is a cycle in a cycle decomposition of $S$ then  $C$ is a minimum $ab$-separating cycle in $G \snip (S\setminus C)$ for otherwise there would be a cheaper minimum $ab$-separating subgraph $G$.  The following lemma allows us to reduce the problem of finding $S$ to that of finding $C$.  That is, it allows us to find $S \setminus C$, and, in particular to do so without specifying the faces we wish to separate.

\begin{lemma}
\label{lem:tightComponentsInMinCut}
Let $S$ be the minimum $ab$-separating subgraph and let $C_1, C_2, \ldots, C_t$ be a cycle decomposition of $S$ with cycles ordered by increasing cost.  Then $C_i$ is the cheapest cycle having $\Z_2$-homology signature $[C_i]$ for $i = 1, \ldots, t-1$.  Moreover, for any $1\leq h\leq t$, $\bigcup_{j=h}^{t}{C_j}$ is a minimum $ab$-separating subgraph in $G\snip \bigcup_{i=1}^{h-1}{C_i}$.
\end{lemma}

\begin{proof}
  If $t = 1$, the lemma is trivially true.  Herein, assume $t \ge 2$.
  For a contradiction, let $C_i$ be the first cycle such that $C_i$ is
  not the cheapest cycle having its homology signature (with $i < t$).
  Let $C_i'$ be the cheapest cycle such that $[C_i'] = [C_i]$.
  We have:
  \begin{equation}
    w(C_i') < w(C_i) \le w(C_{i+1}) \le w(S
    \setminus C_i)\label{eq:wci}
  \end{equation}

  Since $S$ is null-homologous, $[C'_i] = [C_i] = [S \setminus C_i]$.  
  So, both $C'_i  \oplus C_i$ and $C'_i \oplus (S \setminus C_i)$ are separating.
  Since  $[S]_0^{ab} = 1 = [C_i]_0^{ab} \oplus[S\setminus C_i]_0^{ab}$ exactly one of $[C_i]_0^{ab}$ or $[S\setminus C_i]_0^{ab}$ $=1$.
  That is, exactly one of $[C'_i]_0^{ab}\oplus [C_i]_0^{ab}$ or $[C'_i]^{ab}\oplus [S\setminus C_i]_0^{ab}$ $=1$.  Thus, either $C'_i
  \oplus C_i$ or $C'_i \oplus (S \setminus C_i)$ is $ab$-separating. 
  By \autoref{eq:wci}, both  $w(C'_i \oplus C_i) < w(S)$ and $w(C'_i \oplus (S \setminus C_i)) < w(S)$.
  This contradicts that $S$ is the minimum $ab$-separating subgraph.
  
  Since $S$ is $ab$-separating in $G$, $H = \bigcup_{j=h}^{t}{C_j}$ is $ab$-separating in $G' = G\snip \bigcup_{i=1}^{h-1}{C_i}$.
  For a contradiction, suppose that $H' \neq H$ is the minimum $ab$-separating cycle in $G'$, that is $|H'| < |H|$.
  It follows that $S' = (\bigcup_{i=1}^{h-1}{C_i})\cup H'$ is $ab$-separating in $G$, and $|S'| < |S|$, contradicting that $S$ is the minimum $ab$-separating subgraph in $G$.
\end{proof}

\newcommand{\doubleTorusfigure}{
\begin{figure}[h]
  \centering
    \includegraphics[height=1.15in]{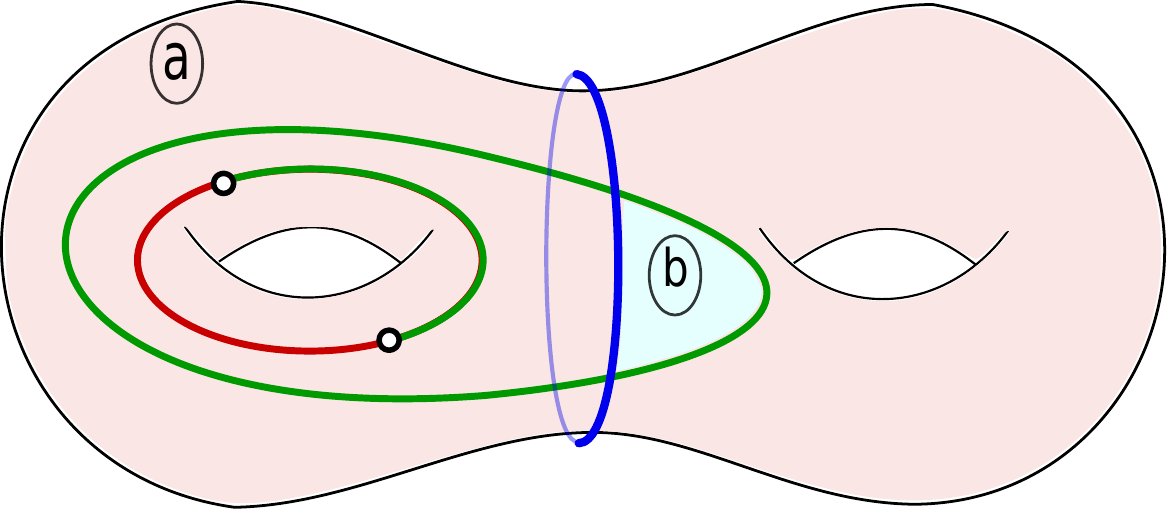}
  \caption{The proof of \autoref{lem:tightSameSide}.  The blue cycle corresponds to $S$, the red path corresponds to $A$, and the green subgraph corresponds to $H$ (shown crossing $S$ in contradiction to the Lemma).  The light region is $F_b \cap J_b$. }
  \label{fig:tightSameSide}
\end{figure}}

The following lemma is stronger than Lemma 6.1 of Erickson et al.~\cite{EFN12}, but the proof technique is similar.  

\begin{lemma}
\label{lem:tightSameSide}
Let $S$ be the minimum $ab$-separating subgraph.
Let $A$ be any subgraph that does not cross $S$, and let $H$ be the minimum subgraph such that $A\oplus H$ is null-homologous.  Then $H$ does not cross $S$.
\end{lemma}
\ifFull\begin{proof} 
The even subgraph $S$ separates the faces of $G$ into two sets $F_a$ and $F_b$.  
Without loss of generality, we assume that $A$ and $a$ are on the same side of $S$ (i.e.\ contained in the same piece of $G \snip S$), $a\in F_a$ and $b\in F_b$.

\doubleTorusfigure

$A \oplus H$ is a null-homologous even subgraph, so it separates the faces of $G$ into subsets $J_b$ and $J'_b$.  Assume, without loss of generality, that $b\in J_b$.  
Let $F = F_b\cap J_b$ and note that $a\notin F$ and $b\in F$, so $\partial(F)$ is $ab$-separating.  
Also, each edge of $\partial(F)$ is on the boundary of a face of $F_b$, so it does not cross $S$.  Finally, $b\in F$, therefore, $\partial(F)$ and $b$ are on the same side of $S$.
In the rest of the proof, we show that $\partial(F)$ must be identical to $S$, thus $H$ does not cross $S$.

Since $S$ is the minimum $ab$-separating subgraph,
\begin{equation}
\label{eqn:FgreaterThanC}
w(\partial F) \geq w(S).
\end{equation}

Each boundary edge of $F = F_b\cap J_b$ should be a boundary edge in at least one of $F_b$ and $J_b$.  Additionally, no boundary edge of $F$ belongs to $A\backslash S$ as $A$ and $b$ (so $A$ and $\partial(F)$) are on different sides of $S$.  Thus, $\partial F \subseteq H\cup S$.

Let $H' = H \oplus \partial F \oplus S$.  Note that $H' \subset H \cup S$ (since $\partial F \subseteq H\cup S$).  Because $\partial F$ and $S$ are boundaries of sets of faces, they are both null-homologous, so $[H'] = [H]$.  Since $H$ is minimum,
\begin{equation}
\label{eqn:deltaPrimeGreaterThanDelta}
w(H') \geq w(H).
\end{equation}

Finally, we show the following inequality by bounding the contribution of each edge $e \in H \cup S$ to the sides of the inequality.
\begin{equation}
\label{eqn:three}
w(\partial F) + w(H') \leq w(S) + w(H).
\end{equation}
\begin{itemize} [noitemsep,nolistsep]
\item If $e \in S\cap H$, then $e$ contributes $2w(e)$ to the right side.  Since, by the construction, at most one copy of each edge is included in each of $\partial F$ and $H'$, $e$ can contribute at most $2w(e)$ to the left side.
\item If $e \in S\oplus H$, then $e$ is in exactly one of $\partial F$ and $H'$ by the definition of $H'$ ($H' = \partial F \oplus H \oplus S$).  In this case, $e$ contributes exactly $w(e)$ to both sides of the inequality. 
\end{itemize}

Therefore, all Inequalities~\eqref{eqn:FgreaterThanC},~\eqref{eqn:deltaPrimeGreaterThanDelta}~and~\eqref{eqn:three} must be equalities.  
In particular, $w(\partial F) = w(S)$.  Thus, the uniqueness of the minimum cut (in the dual graph) implies that $\partial F$ and $S$ are identical, which in turn implies that $H$ does not cross $S$.
\end{proof}
\fi

\begin{lemma}
\label{lem:pathSameSide}
Let $S$ be a minimum $ab$-separating subgraph.  Let $P$ be an $x$-to-$y$ path that does not cross $S$. 
Let $H$ be the shortest subgraph such that $H \oplus P$ is a null-homologous even subgraph.  $H$ contains a tight $x$-to-$y$ path $P'$ that does not cross $S$.
\end{lemma}

\begin{proof}
For $H\oplus P$ to be an even subgraph, all vertices of $H$ except $x$ and $y$ must be even degree and $x$ and $y$ must have odd degree.  
Therefore, $x$ and $y$ must be in the same connected component $C_{xy}$ of $H\oplus P$.
Further, $C_{xy}$ has an Eulerian $x$-to-$y$ walk (i.e., a walk that uses every edge exactly once), and so $C_{xy}$ can be decomposed into an $x$-to-$y$ path $P'$ and a set of cycles ${\cal C}_{xy}$ by iteratively removing cycles from the walk.  

Since $H\oplus P$ is tight, $P'$ must be a tight path and all the cycles in ${\cal C}_{xy}$ must be tight, for otherwise, one could replace $P'$ or a cycle in ${\cal C}_{xy}$ with a cheaper path or cycle while not changing $[H \oplus P]$.

By \autoref{lem:tightSameSide}, $H$ does not cross $S$, and so in particular, $P'$ does not cross $S$, giving this lemma.
%
\end{proof}

\begin{figure}[h]
  \centering
    \includegraphics[height=1.15in]{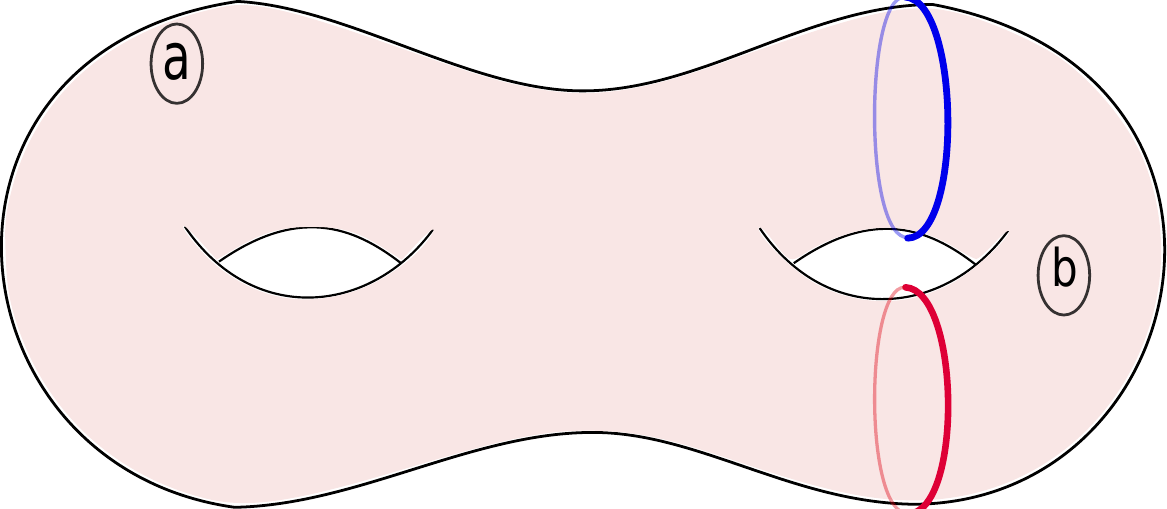}
  \caption{Let {\color{blue}$C$} be the blue cycle. The red cycle is $ab$-separating in $G$ cut along $C$ but not in $G$.}
  \label{fig:tightSameSide}
\end{figure}

Although cutting along a cycle $C$ that does not cross a minimum $ab$-separating subgraph $S$ reduces the genus of the surface, $S$ may not be a minimum $ab$-separating subgraph in $G \snip C$ since the minimum $ab$-separating subgraph in $G \snip C$ may separate copies of $C$ and so may not be separating in $G$.  See Figure~\ref{fig:tightSameSide} for an illustration. To overcome this, we use \autoref{lem:pathSameSide} to witness a tight path $P$ connecting the two copies of $C$ in $G\snip C$.  $G \snip (P \cup C)$ then has one boundary formed by two copies of each edge in $P \cup C$.  By Claim~\ref{clm:sep}, a separating subgraph in $G \snip (P \cup C)$ will also be separating in $G$.

\begin{lemma}
\label{lem:TightCycleandPath}
Let $G$ be a graph embedded on a surface with $b$ boundaries and genus $g$.
Let $S$ be the minimum $ab$-separating subgraph of $G$.  If $S$ is composed of at most $g$ cycles then
there is a tight cycle $C$ and a tight $C$-to-$C$ path $P$ such that 
\begin{enumerate}
\item [(1)] $G \snip (C \cup P)$ is embedded on a surface of genus $g-1$ with $b+1$ boundaries and 
\item [(2)] neither $C$ nor $P$ crosses $S$.
\end{enumerate}
Moreover, $S$ is the minimum $ab$-separating subgraph in $G\snip(C\cup P)$.
\end{lemma}

\begin{proof}
  Since $S$ has at most $g$ cycles, there is a non-separating cycle $C'$ in $G \snip S$.  Let $H'$ be the shortest even subgraph that is homologous to $C'$, and let $C$ be any cycle in the cycle decomposition of $H'$.  By \autoref{lem:tightSameSide}, $H'$ does not cross $S$, so, in particular, $C$ does not cross $S$.  Since $C$ is non-separating, $G \snip C$ has genus one less than $G$, or $g-1$.  Taking the copies of $C$ in $G \snip C$ to be boundaries, $G \snip C$ has $b+2$ boundaries.

Since $C$ is non-separating for $G\snip S$, there exists a path $P'$ between the two copies of $C$ in $G \snip C$ that does not cross $S$. Let $H$ be the minimum  subgraph such that $P'\oplus H$ is null-homologous.  By Lemma~\ref{lem:pathSameSide}, $H$ contains a tight path $P$ that does not cross $S$ and connects the two copies of $C$ in $G \snip C$. Then $G \snip C \snip P$ has genus $g-1$ and $b+1$ boundaries.

Since $S$ does not cross $C$ and $P$, it is $ab$-separating in $G' = G\snip(C\cup P)$.  
For a contradiction, suppose that $S'\neq S$ is the minimum $ab$-separating subgraph in $G'$, that is $|S'|<|S|$.  By Claim~\ref{clm:sep}, $S'$ is $ab$-separating in $G$, which contradicts the assumption that $S$ is the minimum $ab$-separating subgraph in $G$.
\end{proof}

We call the cycle and path described by Lemma~\ref{lem:TightCycleandPath} a {\em tight cycle-path pair}.
The following lemma ensures the possibility of reducing the genus of a graph, while looking for a minimum $ab$-separating subgraph, by cutting a long either a tight cycle or a tight cycle path pair.

\begin{lemma}
\label{lem:recurse}
Let $S$ be a minimum $ab$-separating subgraph in a graph $G$ embedded in a surface of genus $g$.  If $g\geq 1$, at least one of the following conditions holds.
\begin{enumerate}
\item [(1)]  There is a tight cycle $C \subsetneq S$. In this case, $S\backslash C$ is a minimum $ab$-separating subgraph in $G\snip C$.
\item [(2)] There is a tight cycle-path pair  $(C,P)$ in $G$ such that $C$ and $P$ do not cross $S$.  In this case, $S$ is a minimum $ab$-separating subgraph in $G\snip(C\cup P)$.
\end{enumerate}
\end{lemma}
\begin{proof}
Let $C_1, C_2, \ldots, C_t$ be the cycle decomposition of $S$.  
If $t>1$, by Lemma~\ref{lem:tightComponentsInMinCut}, $C_1$ is tight and $S\backslash C_1$ is a minimum $ab$-separating subgraph in $S\snip C_1$.
Otherwise, if $t = 1$, because $g\geq 1$, Lemma~\ref{lem:TightCycleandPath} implies the existence of a tight cycle $C$ and a tight $C$-to-$C$ path $P$ (in $G\snip C$) that do not cross $S$.  Further, the same lemma ensures that $S$ is a minimum $ab$-separating subgraph in $G\snip (C\cup P)$.
\end{proof}

\subsection{A collection of planar problems}

We recursively use Lemma~\ref{lem:recurse} to construct a set of planar graphs, each annotated with a set of non-separating cycles of the original graph, that collectively contain the minimum $ab$-separating subgraphs for all pair of faces, $a$ and $b$.  Starting with graph $G$ of genus $g$, we create a set of new graphs with genus $g-1$, each obtained by either cutting along a tight cycle $C$ that will belong to the separating subgraphs in the derived graphs (as per Lemma~\ref{lem:recurse}, part~(1)) or by cutting along a tight cycle-path pair that will not cross the separating subgraphs of the derived graphs (as per Lemma~\ref{lem:recurse}, part~(2)).

We use $\ssg(H, a, b)$ and $|\ssg(H, a, b)|$ to refer to the minimum $ab$-separating subgraph in $H$ and its weight, respectively.
Let $\mathcal{C}$ be the set of all tight cycles in $G$.  Note that $|\mathcal{C}| = 2^{2g}$ by the definition of $\mathds{Z}_2$-homology.  Let $\mathcal{CP}$ be the set of all tight cycle-path pairs.  
For each tight cycle $C$ and each homology class $h$, $\mathcal{CP}$ contains one pair $(C,P)$, where $P$ is the shortest path in homology class $h$ that connects copies of $C$ in $G\snip C$.  Therefore, $|\mathcal{CP}| = 2^{2g} \times 2^{2g} = 2^{4g}$.  We find the following lemma helpful for computing $\cal C$ and $\cal CP$.

\begin{lemma}[Erickson and Nayyeri~\cite{EN11}, Theorem 6.3]
\label{lem:EN11}
Let $G$ be an undirected graph with nonnegative
edge weights, cellularly embedded on a surface
of genus $g$ with $b$ boundary components. A minimum-weight
even subgraph of $G$ in every $\mathbb{Z}_2$-homology class can be
computed in $2^{O(g+b)}n\log n$ time.
\end{lemma}

\begin{lemma}
\label{lem:tightCycleComp}
The sets $\cal C$ and $\cal CP$ can be computed in $2^{O(g)}n \log n$ time.
\end{lemma}
\begin{proof}
By Lemma~\ref{lem:EN11}, $\cal C$ can be computed in $2^{O(g)}n\log n$ time.  
For each $C\in \cal C$, we apply Lemma~\ref{lem:EN11} to $G\snip C$ to compute all tight paths that pair with $C$ in $\cal CP$ in $2^{O(g)}n \log n$ time.
Therefore, in $2^{2g}\times 2^{O(g)}n\log n = 2^{O(g)}n\log n$ time, we can compute $\cal CP$, too.
\end{proof}

In the following, the positive real numbers are the annotations to the planar graphs that correspond to the set of tight cycles that are cut along that should belong to the separating subgraphs (as per Lemma~\ref{lem:recurse}, part~(1)).
\begin{lemma}
\label{lem:planarCollection}
Let $G$ be a graph embedded in a surface of genus $g$.  There exist a set $\cal H$  of at most $2^{2g^2}$ planar graphs, each annotated with a set of at most $g$ cycles, such that:
\begin{enumerate}
\item [(1)] For any $(H, {\cal C}) \in \cal H$, $F(H)$ is $F(G)$ plus a set of boundary faces.
\item [(2)] For any $(H, {\cal C}) \in \cal H$, for any $a,b\in F(G)$, $|\ssg(H, a,b)| + |{\cal C}| \geq |\ssg(G, a,b)|$.
\item [(3)] For any $a,b\in F(G)$ there exists $(H,{\cal C}) \in \cal H$ such that $\ssg(H, a,b) \cup {\cal C} = \ssg(G, a,b)$.
\end{enumerate} 
Moreover, $\cal H$ can be computed in $2^{O(g^2)}n\log n$ time.
\end{lemma}
\begin{proof}
We use induction on the genus of $G$.  For $g=0$, ${\cal H} = \{(G,\emptyset)\}$ and the properties of the lemma trivially hold.

Let ${\cal C}_G$ and ${\cal CP}_G$ be the set of tight cycles and tight cycle-path pairs for graph $G$ of genus $g$.  For any $C \in {\cal C}_G$, $G \snip C$ has genus $g-1$; let ${\cal H}_{G \snip C}$ be the set of at most $2^{2(g-1)^2}$ annotated planar graphs guaranteed by the inductive hypothesis for the graph $G \snip C$. For any $(C,P) \in {\cal CP}_G$, $G \snip (C \cup P)$ has genus $g-1$; let ${\cal H}_{G \snip (C\cup P)}$ be the set of at most $2^{2(g-1)^2}$  annotated planar graphs guaranteed by the inductive hypothesis for the graph $G \snip (C\cup P)$.  Let
\[
{\cal H} =\left(\bigcup_{C\in \mathcal{C}_G} \left\{(H,{\cal C} \cup C) \, :\, (H,{\cal C}) \in {\cal H}_{G \snip C} \right\}\right) \cup \left(\bigcup_{(C,P)\in\mathcal{CP}_G}{{\cal H}_{G \snip (C\cup P)}} \right)
\]
Since $|{\cal H}_{G \snip C}| \le 2^{2(g-1)^2}$ and $|{\cal H}_{G \snip (C\cup P)}| \le 2^{2(g-1)^2}$ by the inductive hypothesis, \[|{\cal H}| = 2^{2g}\times 2^{2(g-1)^2} + 2^{4g}\times 2^{2(g-1)^2} \le 2 \times 2^{4g}\times 2^{2(g-1)^2} \le 2^{2g^2}.\]
Let $T(n,g)$ be the running time of our algorithm to compute $\cal H$ (for a graph of genus $g$ with $n$ faces). Lemma~\ref{lem:tightCycleComp} implies that $\cal C$ and $\cal CP$ can be computed in $2^{O(g)}n\log n$ time.  Since $G \snip C$ has two more faces that $G$ and $G \snip (C \cup P)$ than $G$,
\[
T(n,g) \leq 2^{O(g)}T(n+2, g-1) + 2^{O(g)} n \log n = 2^{O(g^2)}n\log n.
\]

We show that $\cal H$ satisfies the remaining properties of the lemma.

\paragraph{Property (1)} For any $C$ (resp.~$(C,P)$) we have $F(G\snip C)\supseteq F(G)$ (resp.~$F(G\snip (C\cup P))\supseteq F(G)$ since cutting along $C$ (resp.~$C\cup P$) only adds boundary faces to $G$.  Thus, this property holds by induction.  

\paragraph{Property (2)} Let $a,b \in F(G)$.  We prove that Property~(2) holds for the annotated graphs derived from $ {\cal H}_{G \snip C}$ and from ${\cal H}_{G \snip (C \cup P)}$ separately.

Consider any $C\in {\cal C}$.  
If $S$ is an $ab$-separating subgraph in $G\snip C$ then $S\cup C$ is $ab$-separating in $G$ by Claim~\ref{clm:sep}.
In particular, the minimum $ab$-separating subgraph in $G$ has length at most $|\ssg(G\snip C, a,b)| + |C|$: $|\ssg(G\snip C, a,b)| + |C| \ge |\ssg(G,a,b)|$.  By induction, we have that for any $(H, {\cal C}) \in {\cal H}_{G \snip C}$, $|\ssg(H,a,b)| + |{\cal C}| \ge |\ssg(G\snip C, a,b)|$.  Therefore, since $(H,{\cal C} \cup C) \in {\cal H}$, and by combining the previous two inequalities gives
$|\ssg(H,a,b)| + |{\cal C} \cup C| = |\ssg(H,a,b)| + |{\cal C}| + |C| \ge |\ssg(G\snip C, a,b)| + |C| \ge |\ssg(G,a,b)|$.

Likewise, for any $(C,P)\in {\cal CP}$, if $S$ is an $ab$-separating subgraph in $G\snip (C\cup P)$ then it must be $ab$-separating in $G$, also by Claim~\ref{clm:sep}.
In particular, the minimum $ab$-separating subgraph in $G$ is not heavier than $S$:  $|\ssg(G\snip (C\cup P), a,b)| \ge |\ssg(G,a,b)|$. By the inductive hypothesis, for any $(H, {\cal C}) \in {\cal H}_{G \snip (C \cup P)}$, we have $\ssg(H, a,b) + |{\cal C}| \geq \ssg(G\snip (C\cup P), a,b)$.  Since $(H,{\cal C}) \in {\cal H}$ and by combining these two inequalities we have $\ssg(H, a,b) + |{\cal C}| \geq|\ssg(G,a,b)|$.

\paragraph{Property (3)} Let $S$ be the minimum $ab$-separating subgraph in $G$.  At least one of the conditions of Lemma~\ref{lem:recurse} must hold.  

If the first condition of Lemma~\ref{lem:recurse} holds, then there is a tight cycle $C\in \mathcal{C}$ such that $S\backslash C$ is a minimum $ab$-separating subgraph in $G\snip C$.  By the induction hypothesis, there is an annotated planar graph $(H,{\cal C})\in {\cal H}_{G \snip C}$ such that $\ssg(H, a, b) \cup {\cal C} = \ssg(G\snip C, a,b) = S\backslash C$.  Since $(H,{\cal C} \cup C) \in  {\cal H}$ and $\ssg(H, a, b) \cup {\cal C} \cup C = \ssg(G\snip C, a,b) \cup C = S = \ssg(G, a,b)$, we achieve the desired property.

If the second condition of Lemma~\ref{lem:recurse} holds, then there is a tight cycle-path pair $(C,P)\in {\cal CP}$ such that $S$ is a minimum $ab$-separating subgraph in $G\snip(C\cup P)$.  By the induction hypothsis, there is an annotated planar graph $(H,{\cal C})\in {\cal H}_{G \snip (C\cup P}$
 such that $\ssg(H, a, b) \cup {\cal C} = \ssg(G\snip (C\cup P), a,b)$.  Since $(H,{\cal C})\in {\cal H}$ and $\ssg(H, a, b) \cup {\cal C} = \ssg(G\snip (C\cup P), a,b)= S= \ssg(G, a,b)$, we achieve the desired property.
\end{proof}

\subsection{The algorithm}

We compute the set of annotated planar graphs $\cal H$ as per Lemma~\ref{lem:planarCollection}, and then for each $(H,{\cal C}) \in {\cal H}$, we solve the all-pairs separating cycle problem in the graph $H$; note that it is only necessary to compute the minimum separating cycle for all pairs of non-boundary faces of $H$.  We represent the set of minimum separating cycles of $H$ using a minimum cut-tree $T_H$ for the dual of the graph $H$.  Since $H$ is planar, we can compute $T_H$ in $O(n \log^3 n)$ time using the planar algorithm (Section~\ref{sec:planar-speed-up}).  We then increase the weight of each edge in $T_H$ by $|{\cal C}|$.  

Let ${\cal T}$ be the collection of all such cut trees. Note that $|{\cal T}| = |{\cal H}| = 2^{O(g^2)}$.  All cut trees in $\cal T$ can be computed in $2^{O(g^2)}n \log^3 n$ time, and can be stored using $2^{O(g^2)}n$ space.  For faces $a$ and $b$ of $G$, $a$ and $b$ are non-boundary faces of $H$ (Lemma~\ref{lem:planarCollection}, Property~(1)).  The minimum $ab$-separating cycle in $H$ corresponds to the minimum $a^*b^*$-cut represented by $T_H$, and can be determined in $O(1)$ time using a Cartesian tree representation of $T_H$.

By Lemma~\ref{lem:planarCollection} (Property~(2) and ~(3) and the weight added to each edge of $T_H$), the weight of the minimum $ab$-separating cycle is given by
\[
\min_{T\in {\cal T}}{w(\text{minimum $a^*b^*$-cut in $T$})}.
\]
Therefore, the minimum $ab$-separating cycle can be determined in $O(|{\cal T}|) = O(|{\cal H}|) = 2^{O(g^2)}$ time.

In \ifFull Section~\ref{sec:merge}, \else the full version of the paper,\fi we show how to merge $k$ cut trees to preserve minimality in time $O(kn \log^2 n)$ time. That is, we show how to, from the set ${\cal T}$, compute a single min-cut tree $T$ for minimum separating subgraphs in $G$ (or minimum cuts in $G^*$).  The total time to compute $T$ by this method is $2^{O(g^2)}(n \log^2 n)$ plus the time to computer ${\cal T}$, or $2^{O(g^2)}(n \log^3 n)$.  This will give:
\begin{theorem}
Let $G$ be a graph embedded on a surface of genus $g$.  The Gomory-Hu tree of $G$ can be computed in $2^{O(g^2)}n \log^3 n$ time.
\end{theorem}

\ifFull
\section{Merging cut trees}
\label{sec:merge}

In this section we show how to merge several cut trees into one tree that preserves minimum cuts.  Formally, given $k$ cut-trees over a common vertex set $V$, ${\cal T} = \{T_1, T_2, \ldots, T_k\}$, into a new cut-tree $T_0$ such that for any $x,y \in V$, the minimum $xy$-cut in $T_0$ is the minimum over $i = 1, \ldots, k$ of the minimum $xy$-cut in $T_i$.  We require that for any $a,b,x,y \in V$, if the minimum $ab$-cut is in $T_i$ and the minimum $xy$-cut is in $T_j$, these cuts do not cross.  We will refer to the cuts in $\cal T$ simply as cuts and we will refer to an $xy$-cut that is a minimum $xy$-cut for some $T_i$ as a minimum cut.

We will represent each of the cut trees to be merged as well as the partial cut tree constructed from the merger by region trees. Using a dynamic tree  data structure~\cite{SleTar-JCSS-83} to implement $R$ gives us $O(\log n)$ time cut and link operations.  By annotating the internal edges of $R$ with the cost of the corresponding cut, we can use the minimum operation to (dynamically) find the minimum cut between any two vertices in $R$ in $O(\log n)$ time as well using a path operation.

\bigskip

We simulate the Gomory--Hu algorithm, using the standard divide and conquer approach, starting with $R_0$ being the star tree whose leaves correspond to the elements of the vertex set $V$.  For $i = 1, \ldots, k$, $R_i$ is the region tree corresponding to $T_i$.  Our assumption that the minimum cuts represented by $\cal T$ do not cross guarantees that any cut we add to $R_0$ will nest with the cuts already represented by $R_0$.  Every time we add a cut to $R_0$ we will split $R_i$ for every $i = 1,\ldots, k$ recurse on the respective halves.  Although the {\em minimum} cuts represented by $R_i$ do not cross the {\em minimum} cuts represented by $R_j$, a minimum cut represented by $R_i$ may cross (non-minimum) cuts represented by $R_j$, these cuts will be lost in the representation, but since any cut that crosses a minimum cut that we have found cannot be a minimum cut, by our non-crossing assumption, our merge procedure will preserve minimum cuts.

\subsection{Finding and representing non-crossing cuts} \label{sec:find}

Given a bipartition $(A,B)$ of the leaves of a region tree $R$, we show how to efficiently find the regions represented by $R$ that nest inside $A$ and that nest inside $B$ in time proportional to the size of the smaller side of $(A,B)$ plus the number of regions of $R$ that cross $(A,B)$.

\begin{lemma}\label{lem:A}
  Let $(A,B)$ be a bipartition of the leaves of a region tree $R$.  Let ${\cal T}_A$ be the set of maximal, dangling subtrees of $R$ whose leaves are all in $A$.  The internal nodes of ${\cal T}_A$ represent all the regions of $R$ that nest inside $A$.
\end{lemma}

\begin{proof}
  Let $v$ be an internal node of $R$ such that the region $S_v \subseteq A$; then the subtree $T_v$ of $R$ rooted at $v$ is a subtree of $R$ whose leaves are all in $A$, so $T_v$ must be a subtree of a tree in ${\cal T}_A$.  Likewise, let $u$ be an internal node of $R$ such that the region $S_u \not\subseteq A$; then the subtree $T_u$ of $R$ rooted at $u$ is a subtree of $R$ whose leaves are not all in $A$, so $T_u$ must is not a subtree of a tree in ${\cal T}_A$.
\end{proof}

We denote by $R[A]$ the region tree created by linking the root-most (in $R$) nodes of ${\cal T}_A$ to a new node and adding a new child node $\beta$ of this node (representing $B$).  $R[A]$ is a region tree for the graph obtained by identifying all the vertices of $B$ (and calling this new vertex $\beta$) and it represents all the regions of $R$ that nest inside $A$.  The root-most vertices of ${\cal T}_A$ can be found in $O(|A|)$ time by a simple search starting from $A$.  Since $|{\cal T}_A| = O(|A|)$, $O(|A|)$ link and cut operations can be used to build $R[A]$.  This gives us:

\begin{corollary} \label{cor:A} 
  $R[A]$ can be found in $O(|A|\log n)$ time.
\end{corollary}

Likewise, we define $R[B]$ symmetrically, built from the set of maximal subtrees ${\cal T}_B$.  Consider the subset $C$ of edges in $R$ that are not in ${\cal T}_A \cup {\cal T}_B$ nor are parent edges of the roots of the trees in ${\cal T}_A \cup {\cal T}_B$.  By definition of ${\cal T}_A$ and ${\cal T}_B$, every edge in $C$ must correspond to a cut of $R$ that crosses $(A,B)$.  $C$ forms a connected subtree of $R$, for the path from any node of $C$ to the root of $R$ must stay within $C$.  We can find $C$ efficiently since $C$ is the minimal subtree of $R$ that spans the parents of the roots in ${\cal T}_A$ for, consider a path $P$ in $R$ connecting two parents of roots in ${\cal T}_A$, no edge in $P$ can belong to a dangling subtree (of, for example, ${\cal T}_B$) nor be the parent edge of such a tree.  Therefore, using a simple search through $R$, we can find $C$.  Using cut and link operations, we can emulate edge contraction in $R$, giving us:
\begin{corollary} \label{cor:B}
  $R[B]$ can be found in $O(|C|\log n)$ time where $|C|$ is the number of cuts represented by $R$ that crosses $(A,B)$.
\end{corollary}

\subsection{Algorithm for merging trees}  

Given the set of complete region trees ${\cal R} = \{R_i\ : \ \forall T_i \in {\cal T}\}$, we initialize $R_0$ to be the star region tree over vertex set $V$.  Let $r$ initially be the root of $R_0$ and let $S = V$.  The input to the recursive method is $r$, and the set of region trees ${\cal R}$.  By construction, we will guarantee that every tree in $\cal R$ has the same leaf set and that these are the vertices that are neighbors of $r$ in $R_0$.

\paragraph{Picking a cut} Let $a$ and $b$ be two children of $r$ in $R_0$ that are also in $V$.  Let $A,B$ be the best of the minimum $ab$-cuts of ${\cal R}$.  Without loss of generality, we assume this cut is witnessed by $R_1$ and that $A$ is the smaller side of the cut.

\paragraph{Updating $\mathbf R_0$}  Cut the elements of $A$ from node $r$ of $R_0$ and make them children of a new node $\alpha$.  Make $\alpha$ a child of $r$.  Relabel $r$ as $\beta$.

\paragraph{Partitioning $\mathbf R_1$} Let $uv$ be the edge of $R_1$ that witnesses the cut $A,B$ where $u$ is the child of $v$ in $R_1$.  Cut $u$ from $v$, creating two trees, and make $u$ the root of the new tree.  Add a child $\alpha$ and $\beta$ to each of $u$ and $v$, representing $A$ and $B$.  These are the two trees $R_1[A]$ and $R_1[B]$ such that the leaves of $R_1[A]$ are $A \cup \{\beta\}$ and the leaves of $R_1[B]$ are $B \cup \{\alpha\}$.

\paragraph{Partitioning $\mathbf R_i$ for $\mathbf i > 1$} Find $R_i[A]$ and $R_i[B]$ as described in \autoref{sec:find}.

\paragraph{Recursion} We recurse on two subproblems, corresponding to $A$ and $B$:
\begin{itemize} [noitemsep,nolistsep]
\item the problem defined by $\alpha$ (as a node of $R_0$), $\{R_i[A]\ : \ i = 1,\ldots,k\}$, and
\item the problem defined by $\beta$ (as a node of $R_0$), $\{R_i[B]\ : \ i = 1,\ldots,k\}$.
\end{itemize}
Note that the neighborhood of $\alpha$ in $R_0$ contains the nodes $A \cup \{\alpha\}$, which are the leaves of $R_i[A]$ for $i = 1,\ldots,k$.  Likewise, the neighborhood of $\beta$ in $R_0$ contains the nodes $B \cup \{\beta\}$, which are the leaves of $R_i[B]$ for $i = 1,\ldots,k$.

\bigskip

\noindent Correctness of this algorithm follows from \autoref{lem:A}: $R_i[A]$ and $R_i[B]$ represent all the cuts represented by $R_i$ that do not cross $A,B$.  Since $A,B$ is a minimum cut, and we assume that minimum  cuts do not cross, the cuts represented by $R_i$ that are represented by neither $R_i[A]$ nor $R_i[B]$ cannot be minimum.

\paragraph{Running time}
In building $R_i[B]$, edges are contracted, but these edges do not appear as edges of $R_2[A]$ or $R_2[B]$ since they correspond to crossing cuts, so the total time spent contracting edges is $O(n \log n)$, using $O(n)$ dynamic tree cut and link operations to implement contractions.
The remaining time for one subproblem takes $O(\min\{|A|,|B|\})$ dynamic tree operations.  Therefore, the total time spent is $T(n) = T(a+1) + T(b+1) + O(k\min\{a,b\}\log n)$ where $a + b = n$.  Solving this recurrence gives $T(n) = O(kn \log^2 n)$.
\fi

\section{Speed-up for planar graphs} \label{sec:planar-speed-up}

Borradaile, Sankowski and Wulff-Nilsen~\cite{BSW10,BSW14} gave an $O(n \log^4 n)$-time {\em cycle-based} algorithm for computing a Gomory--Hu cut tree of a planar graph $G = (V,E)$, assuming that minimum cuts are unique and so any two minimum cuts are guaranteed to nest.  The algorithm is guided by a {\em recursive decomposition} of the graph by small, balanced separators.  Working from leaf-to-root in this recursive decomposition, the algorithm finds all the minimum cuts between unseparated vertices in a {\em piece} of the decomposition by finding the corresponding minimum separating cycles in the dual graph.  The cycles are found by computing shortest paths explicitly within the piece and implicitly outside the piece by relying on precomputed distances outside the piece between all pairs of boundary vertices of the piece (represented by an {\em external dense distance graph}.  Computing one cycle incurs a $\log^3 n$ factor in the runtime; combined with the logarithmic depth of the recursive decomposition results in a $\log^4 n$ factor in the running time.  We overcome this bottleneck by instead computing the maximum flow between each pair of unseparated vertices and then extracting the minimum cut from this flow.   
Similar ideas have been used for flow problems, but not for cut problems~\cite{BKMNW11,LNSW12}.  

In both the cycle- and flow-based algorithms, a partial region tree is updated with each newly found cut.  The running time due to this update, in the original cycle-based algorithm, also met $\log^4 n$-factor bottleneck.  We improve this by using a slightly modified version of the region-tree update step of the cycle-based algorithm.  As in the original, we will assume, without loss of generality, that $G$ is triangulated and has bounded degree.

\subsection{Recursive Decomposition}\label{sec:rd} A \emph{piece} $P$ is a subset of $E$ that we regard as a subgraph of $G$ that inherits its embedding from $G$.  A \emph{boundary vertex} of $P$ is a vertex of $P$ incident in $G$ to a vertex not in $P$ and we let $\delta P$ denote the set of boundary vertices of $P$. A \emph{hole} $H$ of $P$ is a face of $P$ which is not a face of $G$. We sometimes regard $H$ as the subgraph of $G$ contained in $H$. Define the \emph{boundary} of $H$ as $\delta H = \delta P\cap H$.

A \emph{decomposition} of $P$ is a set of sub-pieces of $P$ such that every edge of $P$ belongs to a unique subpiece, except that edges with both endpoints in $\delta P$ may belong to more than one subpiece. A \emph{recursive decomposition} of $G$ is obtained by first finding a decomposition of $G$ and then recursing on each sub-piece until pieces of constant size are obtained. Ancestor/descendant relations between resulting set of pieces $\mathcal P$ are defined by their relations in the recursion tree.  The specific type of recursive decomposition we use can be computed in $O(n \log n)$ time and has the following properties~(Section 6, \cite{BSW14}):
\begin{itemize} [noitemsep,nolistsep]
\item each piece of $\mathcal P$ is connected and has a constant number of holes and a constant number of child pieces
\item $\sum_{P\in\mathcal P} |E(P)| = O(n\log n)$
\item $\sum_{P\in\mathcal P} |\delta P|^2 = O(n\log n)$
\end{itemize}
To simplify the analysis, we shall further assume that each non-leaf piece has exactly two children; generalizing to a constant number of children is straightforward.

\subsection{Dense dual-distance graphs}  \label{sec:ddg}
A dense distance graph is a weighted, complete graph on a subset of vertices of the original graph where the weight of an edge equals the shortest path distance in the original graph.  

In the cycle-based algorithm, computations of recursive decompositions, distances, etc. stayed completely in the dual graph $G^*$.  In our flow-based algorithm, we instead use a recursive decomposition of $G$, and compute flows and cuts in $G$; however, we rely on distances precomputed in $G^*$.  For a piece $P$, rather than computing distances in $G$ between vertices of $\partial P$, we compute distances in $G^*$ between vertices of $G^*$ that correspond to faces of $G$ that are incident to $\partial P$.  Let $F_I^*(P)$ (resp.\ $F_E^*(P)$) denote the vertices of $G^*$ that correspond to faces of $G$ that are incident to $\partial P$ and in $P$ (resp.\ not in $P$).  Since $G$ is triangulated and has bounded degree,  $\sum_{P\in\mathcal P} |F_I^*(P)|^2 = O(n\log n)$ and $\sum_{P\in\mathcal P} |F_E^*(P)|^2 = O(n\log n)$.

The external dense dual-distance graph $\DDG(P)$ for a piece $P$ is the dense distance graph for the vertex set 
$F_E^*(P)$ representing distances in the subgraph of $G^*$ induced by edges that are not in $P$.  Some external distances may not be finite since the complement of $P$ is not necessarily connected; we can represent $\DDG(P)$ instead as a union of dense distance graphs, one corresponding to each component of the complement of $P$.  The set of all external dense dual-distance graphs, $\{\DDG(P)\ : \ P \in {\cal P}\}$, can be computed in $O(n \log^3 n)$ time using minor modifications to an algorithm by \L acki, Nussbaum, Sankowski and Wulff-Nilsen~\cite{LNSW12}; see
\ifFull
the next section
\else
the full version of this paper
\fi
for details.

\ifFull

\newpage
\subsection{Computing dense dual distance graphs}
\label{sec:compute-ddgp}

In this section, we give the details for computing the set of all external dense dual-distance graphs, $\{\DDG(P)\ : \ P \in {\cal P}\}$ in $O(n \log^3 n)$ time.  We will require an auxiliary dense distance graph: the internal dense dual-distance graph $\IDDG(P)$ for a piece $P$ is the dense distance graph for the vertex set 
$F_I^*(P)$ representing distances in the subgraph of $G^*$ induced by edges of $P$ (which is not $P^*$).  $\IDDG(P)$ can be computed in $O(|P| \log |P|)$ time using the multiple-source shortest path algorithm due to Klein~\cite{Klein05}.

Consider the cycles defining the boundaries of holes over all pieces of the recursive decomposition. Using the recursive decomposition in~\cite{LNSW12}, these cycles nest and hence form a laminar family which can be represented by a forest $\mathcal F$ of rooted trees. Refer to holes that are external faces of pieces as \emph{external} and all other holes as \emph{internal}. We first describe how to find dense dual distance graphs for internal holes. Consider one such hole $H$ with a boundary $C$ corresponding to a leaf in $\mathcal F$. Then $\DDG(H)$ is equal to $\IDDG(P)$ for a piece with external face bounded by $C$. The dense dual distance graphs for the remaining internal holes are now found bottom-up in $\mathcal F$. To see how this can be done, let $H$ be one of these holes with a boundary $C$. There is a piece $P$ whose external face is bounded by $C$ and $\DDG(H)$ can be obtained using a fast Dijkstra variant of Fakcharoenphol and Rao~\cite{FR06} on the union of $\IDDG(P)$ and the already computed $\DDG(H')$ for each internal hole $H'$ of $P$. This takes a total of $O(n\log^3n)$ time.

It remains to compute $\DDG(H)$ for each external hole $H$. We start with those bounded by cycles corresponding to roots in $\mathcal F$. The algorithm of Fakcharoenphol and Rao is applied to each such hole in $O(n\log^2n)$ time per hole. We claim that this takes a total of $O(n\log^3n)$ time which will follow if we can bound the number of roots by $O(\log n)$. Observe that each step in the construction of the recursive decomposition consists of splitting a piece in two by a cycle separator and recursing on each side. Traversing the branch of the recursion corresponding to pieces on the outside of each cycle separator, we encounter all cycles corresponding to the roots of $\mathcal F$, giving the desired $O(\log n)$ bound. For the remaining external holes, we compute their dense dual distance graphs top-down in $\mathcal F$. To see how, suppose one of these holes $H$ is about to be processed and let $C$ be its boundary. Let $P$ be a piece with an internal hole $H_C$ with boundary $C$ and let cycle $C'$ bound the external face of $P$. For each hole $H'\neq H_C$ of $P$, $\DDG(H')$ has already been computed if $H'$ is internal. This is also the case if $H'$ is external since $C'$ is an ancestor of $C$ in $\mathcal F$. Furthermore, $\DDG(H)$ can be obtained from the union of $\IDDG(P)$ and $\DDG(H')$ over all such holes $H'$. Again, using the fast Dijkstra variant of Fakcharoenphol and Rao, computing all external dense distance graphs takes $O(n\log^3n)$ time.
\fi

\subsection{Region subpieces}\label{sec:rs}

As in the cycle-based algorithm, our flow-based algorithm {\em processes} pieces of the recursive decomposition in a leaf-to-root order.  Processing a piece $P$ involves separating every pair of unseparated vertices in $P$.  We maintain a region tree as described in the preliminaries.  For a pair of vertices $s$ and $t$ in $P$ that are not yet separated, there is a corresponding region $R$ in the region tree that contains $s$ and $t$.  We focus our attention on a {\em region subpiece} which is $P \cap R$.  Borradaile, Sankowski and Wulff-Nilsen argue that for a leaf-most unprocessed piece $P$ with child pieces $P_1$ and $P_2$, a region subpiece contains at most one pair of unseparated vertices, the number of region subpieces corresponding to $P$ is $O(|\partial P_1 \cup \partial P_2|)$, and that all the region subpieces corresponding to $P$ can be computed in time $O((|P|+ |\partial P_1 \cup \partial P_2|^2)\log^2 n)$, have total size $O(|P|)$ and inherit a total of $O(|\partial P|)$ boundary vertices from $P$.

Given these bounds, in the sequel, we focus on a single region subpiece $R_P$ with unseparated vertices $s$ and $t$.  The boundary vertices $\partial R_P$ of $R_P$ are inherited from $P$.

\subsection{Separating $s$ and $t$} \label{sec:st}

Separating $s$ and $t$ is done by first computing a maximum $st$-flow in $G$ which is explicitly represented on $E(R_P)$ and implicitly represented on $E(G) \setminus E(R_P)$. Given $\DDG(R_P)$, the running time of the algorithm is $O((|R_P| + |\partial R_P|^2)\log^2n)$ which by the properties of the recursive decomposition and the bounds on the region subpieces is $O(n\log^3n)$ over all pieces of the recursive decomposition.  The algorithm is nearly identical to a part of the single-source, all-sinks maximum flows algorithm due to \L acki et~al. (Section III C~\cite{LNSW12}) wherein they compute the flow between two cycles rather than two vertices; since $s$ and $t$ can be regarded as degenerate cycles, we can use the same algorithm.  The main difference is that, in order to update the region tree, we must identify the cut edges corresponding to the maximum $st$-flow which is represented largely implicitly.

In order to explain how we determine the cut edges, it suffices to explain how flows are represented implicitly by \L acki et~al., rather than explain their entire algorithm which we use as a black box.  The flow is given by an explicit flow $f_P$ on each edge of $P$ and a circulation $f_C$ defined in the entire graph. The latter is given by a potential function $\phi$ on the set of faces of $G$ that is updated during the algorithm. The circulation $f_C$ is defined by $f_C(uv) = \phi(f_2) - \phi(f_1)$, where $f_1f_2$ is the (directed) dual edge corresponding to $uv$.  It turns out that that it suffices to maintain $\phi(f)$ for faces $f$ incident to $R_P$; this compact representation has been used in the recursive planar flow algorithms by \L acki et~al.~\cite{LNSW12} and Borradaile et al.~\cite{BKMNW11}.  To see why, consider an edge $fg$ of $\DDG(R_P)$; recall that $f$ and $g$ are faces incident to $\partial R_P$ not in $R_P$.  Edge $fg$ corresponds to a shortest path $Q_{fg}$ in $G^*[E(G) \setminus E(P)]$.  The total flow crossing $Q_{fg}$ is given by the sum of the flow $\phi(u) - \phi(v)$ on each edge $uv$ of $Q$; the sum is telescoping and the total flow crossing $Q_{fg}$ is $\phi(g)-\phi(f)$.  $Q_{fg}$ is saturated by the flow if $\phi(g)-\phi(f)$ is equal to the weight of $fg$ in $\DDG(R_P)$.  Since $\phi$ and $\DDG(R_P)$ are maintained, we can find all such {\em saturated} edges of $\DDG(R_P)$ in time $O(|\partial R_P|^2)$.  Let $\DDG_0(R_P)$ be the subgraph of $\DDG(R_P)$ of saturated edges.

Consider a hole of $P$ with boundary $C$; take the hole to not be the infinite face and order the vertices of $C$ cyclically in a clockwise order.  For two vertices $a$ and $b$ of $R_P \cap C$, there is a residual $a$-to-$b$ path in $G \setminus R_P$ only if there is no edge of $\DDG_0(R_P)$ from a face $f$ incident to a part of $C$ from $a$ to $b$ to a face $g$ incident to a part of $C$ from $b$ to $a$;  the path in $G^*$ corresponding to $fg$ consists of edges that are saturated from the $a$ side on the hole to the $b$ side of the hole. 

\newcommand{\holefigure}{%
\begin{figure}[h]
  \vspace{-6ex}
  \centering \label{fig:hole}
  \includegraphics{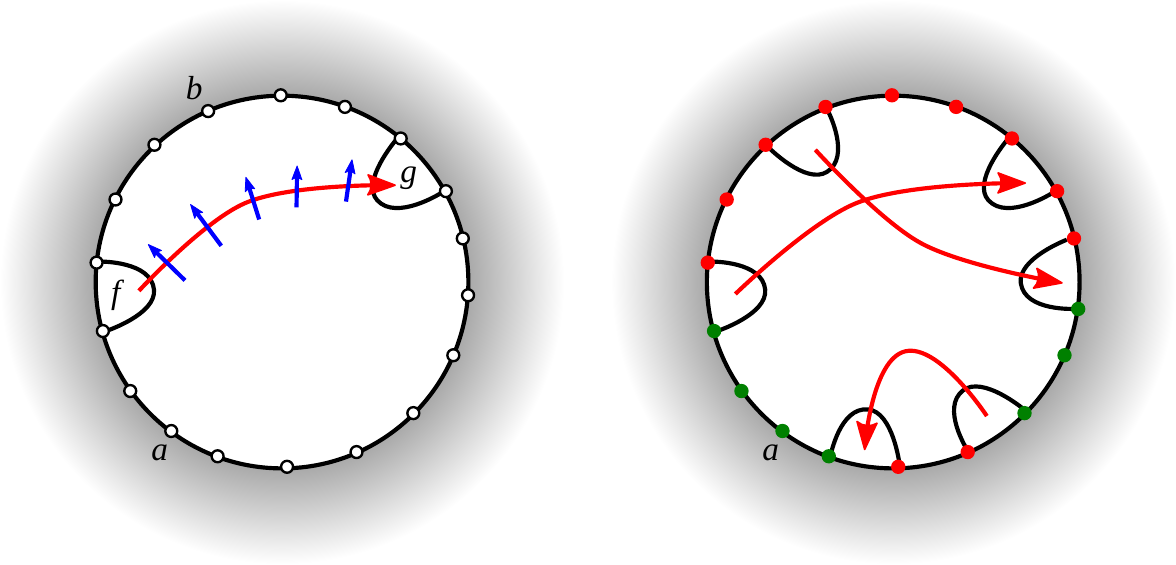}
  \vspace{-6ex}
  \caption{A hole (white) bounded by a cycle and vertices.  The red curves represent edges of $\DDG_0(R_P)$ and correspond to paths in the dual graph within the hole; the corresponding primal edges (blue) are saturated by the flow.  (left) There is no residual path within the hole from $a$ to $b$.  (right) The vertices reachable from $a$ by residual paths within the hole are highlighted green.}
\end{figure}}

\ifFull
\holefigure
\fi

For a vertex $a$ of $R_P\cap C$, we argue that we can determine the subset $S_a$ of $R_P\cap C$ of vertices that are reachable by a residual path in $G \setminus R_P$ in $O(|\partial R_P| \log n)$ time.  Represent the out-neighbors of a face $f$ of $\DDG_0(R_p)$ in clockwise order around $C$.  By binary search, we can determine the last out-neighbor $g$ of $f$ in this order that is on $C$ between $f$ and $a$ or determine that no such edge of $\DDG_0(R_P)$ exists.  This restricts the vertices reachable from $a$ by residual paths in $G \setminus R_P$ to those vertices on $C$ between $g$ and $a$.  By considering each of the faces in order around $C$ starting with the face immediately clockwise of $a$ on $C$, we can, in this way, determine $S_a$ in time $O(|\partial R_P| \log n)$.  Repeating for every vertex of $R_P\cap C$ and for every hole of $P$, we can build an external reachability graph representing reachability via residual paths in $G \setminus R_P$ among vertices of $\partial R_P$ in $O(|\partial R_P|^2 \log n)$ time.

The $s$-side of the cut is given by those vertices reachable by paths that are residual with respect to the flow.  We can find the subset of these vertices in $R_P$  by searching alternately inside $R_P$ via a straightforward search along residual edges (since the flow on edges of $R_P$ are represented explicitly) and search along edges of the external reachability graph in time $O(|R_P| + |\partial R_P|^2)$.  The total time to compute the flow using the algorithm of \L acki et~al., build the external reachability graph and determine the vertices of $R_P$ on the $s$-side of the cut is dominated by the \L acki et~al.\ algorithm, which is  $O((|R_P| + |\partial R_P|^2)\log^2n)$.

\subsection{Updating the region tree}\label{sec:rt}

We will describe how to update the region tree with this cut in terms of the corresponding separating cycle $C$. While performing this update we can, from the $s$-side of the cut within $R_P$, represent the minimum separating cycle by a subset of edges (non-residual edges at the boundary of the search) of $R_P$ and a subset of the edges of $\DDG_0(R_P)$; the total size of this representation is $|C| = O(|R_P|+|\partial R_P|)$ and  can be determined as part of the search. The algorithm we describe is a modified version of that described by Borradaile, Sankowski and Wulff-Nilsen~\cite{BSW14} that achieves a logarithmic speed-up.  We say that an edge $e\in E$ is a \emph{boundary edge} of a region $R$ if $e$ is contained in the bounding cycle of $R$;  Borradaile et~al.\ showed that by maintaining the region tree as a top tree, we can determine in $O(\log n)$ time whether a given edge is a boundary edge of a given region.  We use this fact in our analysis below.

We show how to update the region tree with $C$.  $C$ partitions the children of region $R$ in the region tree.  Let $\mathcal C_R$ be the child regions of $R$ in the region tree and let $\mathcal C_R'$ be the subset contained in the inside of $C$. To correctly update the region tree, we need to remove $\mathcal C_R'$ as children of $R$, add them as children of $C$, and add $C$ as a new child of $R$. If we can identify $\mathcal C_R'$, this update can be done in $O(|\mathcal C_R'|\log n)$ time since there are $O(|\mathcal C_R'|)$ topological changes and each change requires $O(\log n)$ update time in a top-tree representation of the region tree.. If instead we have identified $\mathcal C_R - \mathcal C_R'$ then updating the region tree can be done in $O((|\mathcal C_R - \mathcal C_R'|\log n))$ time. We shall run two algorithms in parallel, one identifying $\mathcal C_R'$, the other identifying $\mathcal C_R - \mathcal C_R'$, and terminate both algorithms when the smaller set has been identified in time $O(\min\{|\mathcal C_R'|, |\mathcal C_R - \mathcal C_R'|\}\log n)$.  For simplicity, we assume that 
 $\mathcal C_R'$ is the smaller set (the other case is symmetric) and show how to find $\mathcal C_R'$.

Let $m_R'$ be the total number of edges (with multiplicity) that bound cycles of regions in $\mathcal C_R'$, excluding edges of $C$.  We will show how to identify $\mathcal C_R'$ in time
\begin{equation}
O(m_R'\log^2n + \log^3 n + |C|).\label{eq:time}
\end{equation}
Over the course of all region tree updates, the second term sums up to $O(n\log^3 n)$ since only a linear number of cycles are added to the region tree.  Likewise, since $|C| = O(|R_P|+|\partial R_P|)$ and $\sum_{P\in \mathcal P} |R_P|+|\partial R_P| = O(\sum_{P\in \mathcal P} |P|) = O(n \log n)$, the third term adds up to $O(n \log n)$.  For the first term, consider distributing this cost among the cycles. Then note that a cycle pays for its edges no more than $O(\log n)$ times. To see this, note that each time a cycle pays, it gets a parent in the region tree with at most $c/2 + 1$ children where $c$ denotes the number of children of the previous parent. This is a constant-factor decrease for $c > 2$. We cannot have $c < 2$ and if $c = 2$ the problem is trivial since the two cycles of $\mathcal C_R'$ must be faces (of constant complexity) of $G^*$ since $C$ separates at least one pair of unseparated faces. Hence, we get a total running time for all region tree updates of $O(n\log^3 n)$.

\subsubsection{Identifying $\mathcal C_R'$}
We identify $\mathcal C_R'$ in two steps. In the first step, we identify those edges that have exactly one endpoint in $C$ and belong to cycles in $\mathcal C_R'$. In the second step, we explore the interior of $C$ starting at these edges to identify the boundaries of all the cycles of $\mathcal C_R'$. We shall only describe the first step as the second step is done exactly as by Borradaile et~al.\ in $O(m_R'\log n)$ time~\cite{BSW14}.
Recall that $C$ is represented by regular edges of a region subpiece together with \emph{super edges} in the external dense dual distance graphs, each representing a path in $G^*$. In the following, we assume that $C$ consists only of super edges as the regular edges are easy to handle using the top tree representation~\cite{BSW14}.

Each super edge $fg$ represents a shortest path in $G^*$ and was found using a fast Dijkstra implementation of Fakcharoenphol and Rao's recursive shortest path algorithm~\cite{FR06} to construct a dense dual distance graph.  In this construction, the path in $G^*$ corresponding to $fg$ has a recursive representation in line with the recursive decomposition of $G$.  That is, $fg$ decomposes into a path $Q_{fg}$ of edges which themselves are super edges in (internal and external) dense dual distance graphs.  The super edges of this path can be recursively decomposed until reaching edges of $G^*$.  The number of recursion levels is at most the depth $O(\log n)$ of the recursive decomposition of $G^*$.  We shall assume that any super edge $ab$ (i) points to the endpoints in $G^*$ of the subpath $Q_{ab}$ that $ab$ decomposes into at the next recursion level, (ii) points to the super edge of $Q_{ab}$ that contains the midpoint of $Q_{ab}$ (as a path in $G^*$), and (iii) is annotated with its length in terms of number of edges of $G^*$ (as well as its length in terms of weights of those edges).   It is easy to maintain this information during the construction of the dense dual distance graphs without an asymptotic time increase.

A minimum separating cycle in a planar graph is an isometric cycle: for any vertex $r$ on $C$, $C$ consists of two shortest paths $Q_1$ and $Q_2$ from $r$ to vertices $a$ and $b$, respectively, together with the single edge $ab$ of $G^*$~\cite{HM94}.  Borradaile, Sankowski and Wulff-Nilsen show how to find an $a$, $b$ and $r$ in time $O(\log^3n + |C|)$ along with a representation of $Q_i$ consisting of at most $|C|$ super edges at the top level of the recursion and $O(\log n)$ paths of super edges from dense dual distance graphs at deeper recursion levels for $i = 1,2$ (Section 4.2~\cite{BSW14}).  For $i = 1,2$, recall that we need to find all edges of $G^*$ with exactly one endpoint on $Q_i$ that are contained in boundary cycles of $\mathcal C_R'$.   By symmetry, we restrict our attention to $Q_1$.

The algorithm is as follows. First, identify the first and last edge (of $G^*$) on $Q_1$. Denote them by $e_1$ and $e_2$, respectively. Then check if $e_1$ and $e_2$ are incident to the same cycle $C'$ of $\mathcal C_R'$. As noted above, each such check can be done in $O(\log n)$ time using a top tree representation of the region tree.

If this check is positive then, since bounding cycles are isometric and shortest paths are unique, all edges of $Q_1$ are in $C'\cap C$ and we are done since none of these edges can be interior to a cycle of $\mathcal C_R'$ since these cycles have disjoint interiors.  (A special case is when a cycle of $\mathcal C_R'$ is a face of $G^*$ which need not be isometric but, since faces have constant size, this will not affect the analysis.)

If this check is negative, we find an edge $e_m$ of $G^*$ belonging to $Q_1$ such that a constant fraction of the (regular) edges of $Q_1$ are before $e_m$ and a constant fraction of them are after $e_m$. We then recurse on the two subpaths of $Q_1$ from $e_1$ to $e_m$ and from $e_m$ to $e_2$, respectively. If there are $k$ cycles of $\mathcal C_R'$ incident to $Q_1$ then the total number of regular edges $e_m$ considered is $O(k\log n)$. This is where our improvement differs from the original region-tree-updating technique of Borradaile et al.~\cite{BSW14}: we show how to identify each edge $e_m$ in $O(\log n)$ time instead of $O(\log^2n)$ time, which was the bottleneck of their method. Total time to identify all $O(k\log n)$ edges is then $O(k\log^2n) = O(m_R'\log^2n)$.

First we need a lemma. For a path $Q$, we denote by $Q[x,y]$ the subpath from $x$ to $y$ and we call a vertex $x$ a \emph{split point} of $Q$ if at most a constant $c < 1$ fraction of the vertices of $Q$ are before resp.~after $x$. \ifFull
\else
The proof is in the full paper.
\fi
\begin{lemma}\label{Lem:SplitPoint}
Let $Q = Q[a,c] \circ Q[c,e] \circ Q[e,f] \circ Q[f,d] \circ Q[d,b]$ be a directed path where possibly some of the subpaths are empty. Let $m$ (resp.~$m(c,d)$) denote the midpoint of $Q$ (resp.~$Q[c,d]$), splitting the path into two (almost) equal-size subpaths such that $|Q[c,d]| > \frac{1}{4}|Q|$, $m\in Q[c,d]$, and $m(c,d)\in Q[e,f]$. We have:
\ifFull
\begin{enumerate} [noitemsep,nolistsep]
\item If $m\in Q[c,e]$ then $e$ is a split point of $Q$.
\item If $m\in Q[f,d]$ then $f$ is a split point of $Q$.
\item If $m\in Q[e,f]$ and $|Q[e,f]|\leq \frac 1 4 |Q|$ then $e$ is a split point of $Q$.
\end{enumerate}
\else
If $m\in Q[c,e]$ then $e$ is a split point of $Q$; if $m\in Q[f,d]$ then $f$ is a split point of $Q$; if $m\in Q[e,f]$ and $|Q[e,f]|\leq \frac 1 4 |Q|$ then $e$ is a split point of $Q$.
\fi
\end{lemma}
\ifFull
\begin{proof}
In the first case, since $m\in [a,e]$, at least half the vertices of $Q$ belong to $Q[a,e]$. Since $Q[m(c,d),d]\subseteq Q[e,b]$, $Q[e,b]$ has length at least $\frac 1 2 |Q[c,d]| - 1 > \frac 1 8 |Q| - 1$. Hence $e$ is a split point of $Q$. The second case is symmetric.
For the third case, since $m\in Q[e,b]$, at least half the vertices of $Q$ belong to this subpath. Since $|Q[e,m]|\leq |Q[e,f]|\leq \frac 1 4 |Q|$, $Q[a,e]$ has length at least $|Q[a,m]| - \frac 1 4 |Q|\geq |Q|/2 - 1 - \frac 1 4 |Q| = \frac 1 4 |Q| - 1$. Hence, $e$ is a split point of $Q$.
\end{proof}
\else
\fi

We use this lemma to find an edge $e_m$ with the property described above. Recall that we store the number of edges of $G^*$ represented by each super edge $ab$ as well as a pointer to the sub-super edge that contains the midpoint of $Q_{ab}$. As described above, path $Q_1$ is represented by $O(|C| + \log n)$ subpaths and since we know the length of each of them, we can in $O(|C| + \log n)$ identify the subpath $Q_1'$ that contains the midpoint of $Q_1$.  $Q_1'$ plays the role of $Q[e,f]$ and $Q_1$ plays the role of $Q[c,d]$ and $Q[a,b]$ in \autoref{Lem:SplitPoint} (so that $a = c$ and $b = d$). The assumptions in the lemma hold and we can check in constant time whether any of the three cases apply. If so, we are done as we have found a split point and we can pick $e_m$ as an edge incident to this point. Otherwise, we have $m\in Q[e,f]$ and $|Q_1'| = |Q[e,f]| > \frac 1 4 |Q|$. The path $Q[e,f]$ belongs to a shortest path of a dense dual distance graph, so in constant time, we can identify the super edge $uv$ of this path that contains its midpoint. Now apply the lemma again but with $Q[e,f]$ defined as $uv$, $Q[c,d]$ defined as $Q_1'$, and $Q[a,b]$ defined as $Q_1$. Then the assumptions in the lemma hold again. Applying it, we either find a split point or we decompose $uv$ into a shortest path of super edges, and recurse on the super edge containing the midpoint of this path. Each step runs in $O(1)$ time and since there are $O(\log n)$ recursion levels, we obtain a split point and hence $e_m$ in $O(\log n)$ time.

We conclude from the above that the total time to add bounding cycles to the region tree is $O(n\log^3 n)$.

\newpage
\paragraph{Acknowledgements} This material is based upon work supported by the National Science Foundation under Grant Nos.\ CCF-0963921, CCF-1228639, and CCF-1252833 and by the Office of Naval Research under Grant No. N00014-08-1-1015.

\bibliographystyle{abbrv}
\bibliography{gh}

\ifFull
\else
\homsigfigure
\doubleTorusfigure
\holefigure
\fi

\end{document}